\documentclass{article}
\usepackage{amsmath,amsthm}
\usepackage{amssymb}
\usepackage{dsfont}
\usepackage{url}
\usepackage{hyperref}
\usepackage{graphicx}

\usepackage{caption}
\usepackage{subcaption}

\newcommand{\Zp}{\ensuremath{\mathds{Z}_+}}
\newcommand{\Pclass}{\textbf{P}}
\newcommand{\NP}{\textbf{NP}}
\newcommand{\NPhard}{{\textbf{NP}-hard}}
\newcommand{\DTIME}{{\textbf{DTIME}}}
\newcommand{\ub}{\ensuremath{IUB}}
\newcommand{\ubof}[1]{\ensuremath{\ub(#1)}}

\newtheorem{theorem}{Theorem}
\newtheorem{lemma}[theorem]{Lemma}

\title{Approximation algorithms for the Transportation Problem with Market Choice and related models}
\author{Karen Aardal, Pierre Le Bodic}

\begin{document}
\maketitle
\begin{abstract}
Given facilities with capacities and clients with penalties and demands, the transportation problem with market choice consists in finding the minimum-cost way to partition the clients into unserved clients, paying the penalties, and into served clients, paying the transportation cost to serve them.

We give polynomial-time reductions from this problem and variants to the (un)capacitated facility location problem, directly yielding approximation algorithms, two with constant factors in the metric case, one with a logarithmic factor in the general case.
\end{abstract}

\section{Introduction}
In the classical transportation problem \cite{kantorovich60a, ahuja93a}, we are given a 
set $F$ of $m$ facilities and a set $C$ of $n$ clients. 
Each facility $i \in F$ has a supply capacity $s_i$ and each client $j \in C$ has a demand $d_j$ (throughout, facilities will always be denoted $i$ and clients $j$, sometimes with a subscript).
The per unit cost of transporting items from facility $i$ to client $j$ is $c_{ij}$.
The transportation problem consists in finding a minimum-cost flow so that each client's demand is met, and such that no supply is exceeded.
Since this problem can be modeled as a special case of the minimum-cost flow problem \cite{ahuja93a}, it can be solved in polynomial time.

The Capacitated Facility Location (CFL) problem \cite{kuehn63a,ahuja93a} is an \NPhard{} generalization of the transportation problem where each facility $i$ has an opening cost $f_i$: using capacity from $i$ requires opening the facility and paying $f_i$. The problem then consists in minimizing the sum of the 
transportation cost and the opening costs.

Damci-Kurt et al. \cite{damcikurt13a} introduced the Transportation problem with Market Choice (TMC), which is also an \NPhard{} generalization of the transportation problem, where a penalty $r_j \in \Zp$ may be paid in exchange for not serving client $j$.
Other classical logistics problems have been studied with additional market choices (see e.g. \cite{geunes11a}).

Problem CFL and its variants have been the subject of an extensive study.
In this article we are mainly interested in approximation results (see e.g. \cite{damcikurt13a} and references therein for exact methods).
Let us first point out that the classical set covering problem \cite{garey79a} is a special case of the uncapacitated variant
with unitary opening costs as follows. 
Suppose -- without loss of generality -- that the set covering instance is feasible.
The given subsets of the universe are represented by facilities, the elements of the universe are clients with demand 1, and the transportation costs are 0 if an element belongs to a subset, and 2 otherwise.
(Since the opening costs are 1, an optimal solution to the instance of CFL never uses edges with cost 2, and thus this solution provides a minimum-size cover).
This directly implies that this variant as well as the more general CFL are strongly \NPhard{} and cannot be approximated within any factor better than $O(\log n)$ unless \Pclass=\NP{}  \cite{alon06a}.
For the uncapacitated variant there actually is an $O(\log n^2 m)$ approximation algorithm \cite{hochbaum82a}, and
for CFL there is an $O(\log n + \log(\max_{i \in F, j \in C}c_{ij}))$ approximation algorithm \cite{barilan01a}.

All constant-factor approximation results presented below consider CFL and its variants in the setting that
the facilities and the clients are points embedded in a common metric space, i.e. the distances between
pairs of points are nonnegative, symmetric and satisfy the following variant of the triangle inequality:
\begin{equation*}
 c_{i_0 j_0} \leq c_{i_0 j_1} + c_{i_1 j_1} + c_{i_1 j_0} 
\end{equation*}
for any two facilities $i_0, i_1 \in F$ and any two clients $j_0, j_1 \in C$, where $c_{ij}$ is the per unit transportation cost between facility $i$ and client $j$.

There is a 5-approximation algorithm for the metric CFL due to Bansal et al.~\cite{bansal12a}, which improves on the work of P\'{a}l et al., Mahdian and P\'{a}l, and Zhang et al.~\cite{pal01a, mahdian03a, zhang05a} (all four articles are based on local-search techniques).
The metric uncapacitated variant (UFL) \cite{cornuejols90a, shmoys00a, mahdian06a} has a $1.488$-approximation algorithm by Li \cite{li13a}, which is based on the work of Byrka and Aardal \cite{byrka10a} and Chudak \cite{chudak03a}, and cannot be approximated within 1.463 unless $\NP \subseteq \DTIME[n^{O(\log\log n)}]$ \cite{guha99a}.
The metric variant of CFL where capacities are uniform is \NPhard{} as well, and the best known approximation algorithm has a factor of 3 and uses local search \cite{aggarwal10a}.
The algorithm was initially given by Kuehn and Hamburger \cite{kuehn63a}, and Korupolu et al.~\cite{korupolu00a} provided the first analysis, while also considering other variants of the problem.
The analysis and the approximation factor was subsequently improved by Chudak and Williamson \cite{chudak05a}.
Non local-search based methods include the recent work of An et al. \cite{an14a}, where the authors formulate an LP relaxation of metric CFL that has a constant integrality gap, and derive an LP-rounding approximation algorithm with factor 288.
Some authors also consider \emph{soft} variants of CFL, where a facility can be opened multiple times.
Currently, the best approximation factor for this problem is 2 \cite{mahdian06a}.

In this article, we establish four polynomial-time reductions preserving approximation factor (see e.g. \cite{papadimitriou88a, ausiello06a}), from TMC to CFL and from CFL to TMC, both under the metric assumption and in the general case.
The reductions are similar in principle and rely on the closeness of the two problems as well as on a good choice of costs per unit of flow in the gadgets we introduce.
Using approximation algorithms established by Bansal et al.~\cite{bansal12a} and by Bar-Ilan et al. \cite{barilan01a}, for the metric case and the general case, respectively, we can then prove that there exists a $5$-approximation algorithms for TMC under the metric hypothesis, and an approximation algorithm with logarithmic factor in the general case.

These two results are given in Sections \ref{sec:5approx} and \ref{sec:logapprox}.
Section \ref{sec:converse} provides the reductions in the other direction, from CFL to TMC.
Finally, Section \ref{sec:uncap} briefly deals with two uncapacitated cases.

\section{Formal problem definition}\label{sec:prelim}
 Throughout the article, we suppose all data is integer and non-negative.
 Let us define problem TMC and CFL, by first giving the input common to both problems.
 Each facility $i \in F$ has a serving capacity $s_i$, i.e. it can send at most $s_i$ units of flow to possibly multiple clients.
 Each client $j \in C$ has a demand $d_j$, that can be satisfied by multiple facilities.
 The per unit transportation cost between facility $i$ and $j$ is denoted by $c_{ij}$.
 In the metric case, we have $c_{ij}=d(i,j)$.

 In TMC, we are additionally given a penalty cost $r_j$ for each client $j \in C$, which must be paid if client $j$ is not served.
 For each facility $i \in F$ and client $j \in C$, let $x_{ij}$ be the flow variable between $i$ and $j$, and for each $j\in C$, let $z_j$ be the binary variable such that $z_j=1$ if and only if client $j$ is not served.
 Problem TMC then consists in minimizing $\sum_{i\in F}\sum_{j\in C}c_{ij}x_{ij}+ \sum_{j \in C}r_jz_j$ such that the demand of each client $j\in C$ is served if $z_j=0$, each facility $i\in F$ sends at most $s_i$ units of flow, and $x$ is a nonnegative flow.

 In CFL, an opening cost $f_i$ is additionally given for each facility $i \in F$, which must be paid if the facility uses any unit of capacity.
 As for TMC, let $x$ denote the flow vector and let $z$ be the binary vector such that for each $i\in F$, $z_i=1$ if and only if facility $i$ is opened.
 Problem CFL then consists in minimizing $\sum_{i\in F}\sum_{j\in C}c_{ij}x_{ij}+ \sum_{i \in F}f_iz_i $ such that the demands of each client $j\in C$ is served, each facility $i\in F$ sends at most $s_i$ units of flow if $z_i=1$, no flow otherwise, and $x$ is a nonnegative flow.
 
 \section{An approximation algorithm with factor 5 for the metric TMC}\label{sec:5approx}
 Essentially, the following lemma shows that not serving a client $j$ in TMC is equivalent to opening a dedicated facility $i$ for this client in CFL, where the opening cost $f_i$ is the penalty cost $r_j$, and the capacity $s_i$ matches the demand $d_j$.
 \begin{lemma}\label{th:reduc-metric-TMC-CFL}
  There is a polynomial-time reduction preserving the approximation factor from TMC to CFL, where both problems are considered under the metric hypothesis.
 \end{lemma}
\begin{proof}
 First, let us describe the polynomial-time reduction.
 Let an instance $I_1$ of TMC be given, and let us build an instance $I_2$ to CFL.
 Initialize $I_2$ with the same data as $I_1$, except that penalties $r_j$ of $I_1$ are not used in $I_2$ and opening costs $f_i$ of $I_2$ are all equal to $0$.
 Furthermore, for each client $j$, create a dummy facility $i$ with opening cost $f_i=r_j$, capacity $s_i=d_j$, serving cost $c_{ij}=0$, and for all clients $j_0\neq j$, set $c_{ij_{0}}=\min_{i_0 \in C\setminus\{i\}}(c_{i_0 j_0}+c_{i_0 j})$ (which satisfies the triangle inequality).
 This reduction runs in linear time in the size of $I_1$.
 Let $opt(I)$ be the optimal value of instance $I$, and let $obj(x,z)$ denote the objective value of a solution $(x,z)$.

 Second, let us prove that for any instance $I_1$, the resulting instance $I_2$ satisfies $opt(I_2)\leq opt(I_1)$.
 Let an optimal solution to $I_1$ be given.
 In the solution that we build for $I_2$, open all facilities with opening cost 0, and use the same flows as in the given optimal solution of $I_1$.
 As a result, some clients in $I_2$ may not be served, namely those for which a penalty is paid in $I_1$.
 For each client not yet served, open the corresponding dummy facility, and serve the client through it.
 The solution is feasible, and the total costs are equal.

 Third, let us prove that for any feasible solution $(x^2,z^2)$ to $I_2$, we can find a solution $(x^1,z^1)$ to $I_1$ in a time polynomial in the size of $I_1$, such that $obj(x^1,z^1)\leq obj(x^2,z^2)$.
 Let us build -- in polynomial time -- a solution $(\hat{x}^2,\hat{z}^2)$ to $I_2$ that satisfies $obj(\hat{x}^2,\hat{z}^2)\leq obj(x^2,z^2)$, and from which we will retrieve $(x^1,z^1)$.
 Open any free facility that is not opened in $(x^2,z^2)$, i.e. set $\hat{z}^2_i=1$ if $f_i=0$, $\hat{z}^2_i=z^2_i$ otherwise.

\begin{figure}
\centering
\begin{subfigure}{.4\columnwidth}
\centering
\includegraphics{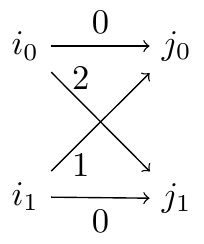}
\caption{Initial flow in which we want to increase the flow between $i_0$ and $j_0$.}
\end{subfigure}
\begin{subfigure}{.15\columnwidth}
\end{subfigure}
\begin{subfigure}{.4\columnwidth}
\centering
\includegraphics{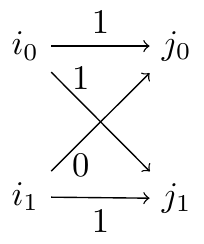}
\caption{Flow after increasing the flow  between $i_0$ and $j_0$ by $\bar{x}=1$.}
\end{subfigure}
\caption{Example with two facilities $i_0$ and $i_1$ and two clients $j_0$ and $j_1$. The amount of flow is indicated on the edges. In both cases, the capacity used by each facility and the demand provided to each client is the same.}
\label{fig:flowswap}
\end{figure}
 Furthermore, we can suppose without loss of generality that in the solution $(\hat{x}^2,\hat{z}^2)$, every dummy facility that is opened fully serves its corresponding client.
 Indeed, if a dummy facility $i_0$ is opened and does not fully serve client $j_0$, then there is at least one other facility, say $i_1$, serving $j_0$.
 If $i_0$ serves no other client, then fully serve client $j_0$ from $i_0$ and delete the flow between $j_0$ and any other facility, at no additional cost.
 If $i_0$ serves a client $j_1$, then let $\bar{x}=\min(\hat{x}^2_{i_0 j_1},\hat{x}^2_{i_1 j_0})$, and increase the flows $\hat{x}^2_{i_0 j_0}$ and $\hat{x}^2_{i_1 j_1}$ by $\bar{x}$ while decreasing the flows $\hat{x}^2_{i_0 j_1}$ and $\hat{x}^2_{i_1 j_0}$ by $\bar{x}$.
 Figure \ref{fig:flowswap} provides an example in this case.
 Following this operation, both facilities $i_0$ and $i_1$ use the same capacity as previously, client $j_0$ and $j_1$ receive the same amount of flow, which ensures feasibility.
 The cost decreases by $\bar{x}(c_{i_0 j_1} + c_{i_1 j_0} - c_{i_0 j_0} - c_{i_1 j_1})$, and we have $c_{i_0 j_0}=0$ on the one hand and $c_{i_1 j_1} \leq c_{i_1 j_0} + c_{i_0 j_1}$ by the triangular inequality on the other hand, thus the expression is non-negative.
 
 Finally, since we can suppose that each opened dummy facility fully serves the corresponding client, remove all dummy facilities and their corresponding clients from the instance and solve the transportation problem on the remainder of the problem.
 Note that for each client, the part of demand supplied by dummy facilities in $(\hat{x}^2,\hat{z}^2)$ is at least what it is in $(x^2,z^2)$,
 therefore on the remainder of the problem the transportation problem is feasible and the cost does not increase, and thus $obj(\hat{x}^2,\hat{z}^2)\le obj(x^2,z^2)$.
 It is now straightforward to build $(x^1,z^1)$ from $(\hat{x}^2,\hat{z}^2)$, and the two solutions have the same cost.
\end{proof}

\begin{theorem}\label{th:approx-TMC}
 There is a 5-approximation for the metric TMC.
\end{theorem}
\begin{proof}
 It follows directly from Lemma \ref{th:reduc-metric-TMC-CFL} and from the existence of an approximation algorithm with factor 5 for CFL in the metric case (Theorem 1 of \cite{bansal12a}).
\end{proof}

Note that Lemma \ref{th:reduc-metric-TMC-CFL} can readily be adapted to reduce the Capacitated Facility Location problem with Market Choice (CFLMC) to CFL.
CFLMC is the generalization of both TMC and CFL, i.e. facilities have opening costs and clients have penalties.
It thus appears that CFLMC is not harder to solve as CFL.
The case of the uncapacitated variant is dealt with separately in Section \ref{sec:UFLMC}.

\section{A logarithmic approximation factor for TMC in the general case}\label{sec:logapprox}
It is possible to adapt Lemma \ref{th:reduc-metric-TMC-CFL} to the general, non-metric case.
\begin{lemma} \label{th:reduc-general-TMC-CFL}
 There is a polynomial-time reduction preserving the approximation factor from TMC to CFL, where both problems are considered in the general case.
\end{lemma}
\begin{proof}
 The proof is similar to the proof of Lemma \ref{th:reduc-metric-TMC-CFL}.
 The only difference in the first paragraph of the proof is as follows: when adding a dummy facility $i$ for a client $j$, set $c_{ij}$ to 0, but set the costs from $i$ to other clients to the maximum unit cost of instance $I_1$.
 The second paragraph is valid without modification.
 For the third paragraph, nothing changes except the cost decrease analysis: to prove $\bar{x}(c_{i_0 j_1} + c_{i_1 j_0} - c_{i_0 j_0} - c_{i_1 j_1})$ is non-negative, we observe that $c_{i_0 j_0}=0$, and that $c_{i_0 j_1} + c_{i_1 j_0} - c_{i_1 j_1} \geq c_{i_0 j_1} - c_{i_1 j_1} \geq 0$ by the choice of unit cost of the dummy facility $i_0$.
\end{proof}

\begin{theorem}
 There is a $O(\log n + \log(\max_{i \in F, j \in C}c_{ij}))$ approximation algorithm for TMC in the general case.
\end{theorem}
\begin{proof}
Corollary 5.5 of \cite{barilan01a} provides an approximation algorithm with ratio $O(\log n + \log(\max_{i \in F, j \in C}c_{ij}))$ for CFL.
In the reduction of Lemma \ref{th:reduc-general-TMC-CFL}, the maximum cost per unit of flow is preserved, and the number of clients $n$ is unchanged.
\end{proof}

\section{Preserving ratio from CFL to TMC}\label{sec:converse}
We establish reductions from CFL to TMC in the metric and general case.
In a very similar fashion, they rely on the addition of a dummy client for each facility, for which the penalty needs to be payed (in TMC) for the facility to be opened (in CFL).

We will assume throughout the remainder of this section that only feasible instances of CFL are considered.
This is not restrictive, as infeasible instances of CFL can be detected in linear time by checking if the total demand exceeds the total supply.
Note that TMC instances are feasible, as it is always possible to pay all client penalties.

\subsection{Instance upper bound}\label{sec:ub}
 We define an instance upper bound (\ub) of a given problem to be a strict upper bound on the objective value of \emph{any} feasible solution.
 
 An \ub{} of CFL can for example be obtained by adding together all facility opening costs as well as the optimal value of the maximization version of the transportation problem with every facility being opened.
 Finally, add 1 for this to be a strict upper bound.
 
\subsection{The metric case}
\begin{lemma}\label{th:reduc-metric-CFL-TMC}
 There is a polynomial-time reduction preserving the approximation factor from CFL to TMC, where both problems are considered under the metric hypothesis.
\end{lemma}
\begin{proof}
 Let an instance $I_1$ of CFL be given, and let us build an instance $I_2$ to TMC.
 The reduction is similar to the one in the proof of Lemma \ref{th:reduc-metric-TMC-CFL}, and we will only indicate changes in each paragraph.
 In the first paragraph, create instance $I_2$ by adding a dummy client $j$ for each facility $i$, such that the penalty for not serving the dummy client is the opening cost of the  facility (i.e. $r_j=f_i$), instead of the opposite.
 All data are initialized similarly, except that non-dummy clients have a penalty of \ubof{I_1} (see Section \ref{sec:ub}), instead of non-dummy facilities having an opening cost of 0.
 
 The second paragraph is the exact converse.
 Let an optimal solution $I_2$ be given for TMC; in the solution $I_1$ that we build for CFL, we open a facility if and only if in $I_2$, the corresponding dummy client is not served.
 The sum of the facility opening costs in $I_1$ is thus equal to the sum of the penalty costs in $I_2$.
 The flows between facilities and non-dummy clients in $I_2$ can then be used between opened facilities and clients in $I_1$, and therefore $opt(I_1)\leq opt(I_2)$.
 
 Before the third paragraph, we firstly need to make sure that no penalty is paid for non-dummy clients (in the original reduction, it sufficed to open every free facility).
 If the given solution is such that  $obj(x^2,z^2)\geq \ubof{I_1}$, then we are in the case where a penalty is paid for a non-dummy client.
 Replace the solution $(x^2,z^2)$ by the one  given by the transportation problem on the instance $I_2$ where all non-dummy clients are served and all dummy clients are not served.
 Since this corresponds to a feasible solution of $I_1$ (where all facilities are opened), its cost is strictly less than \ubof{I_1}.
 The rest of the proof is then direct, using the triangle inequality to show that a dummy client can be served entirely by its corresponding facility at no additional cost.
\end{proof}

\subsection{The general case}
\begin{lemma}\label{th:reduc-general-CFL-TMC}
 There is a polynomial-time reduction preserving the approximation factor from CFL to TMC, where both problems are considered in the general (i.e. non-metric) case.
\end{lemma}
\begin{proof}
 The proof is based on the proof of Lemma \ref{th:reduc-metric-CFL-TMC}, using the unit costs between dummy clients and facilities as in the proof of Lemma \ref{th:reduc-general-TMC-CFL}.
\end{proof}

\section{Uncapacitated variants}\label{sec:uncap}
We prove that the Uncapacitated Transportation problem with Market Choice (UTMC) has an approximation algorithm with ratio $1.488$ and we discuss (not being able to) extending these results to the Uncapacitated Facility Location problem with Market Choice (UFLMC).

\subsection{Approximation of the Metric Uncapacitated Transportation problem with Market Choice}
UTMC is a special case of TMC where each capacity is greater or equal than the total demand.
UTMC thus also reduces to CFL and can be approximated with a factor 5 (Theorem \ref{th:approx-TMC}).
We can however easily adapt Lemma \ref{th:reduc-metric-TMC-CFL} to the case where both problems are uncapacitated, which yields a better approximation ratio.
\begin{lemma}\label{th:reduc-metric-UTMC-UFL}
 There is a polynomial-time reduction preserving the approximation factor from UTMC to UFL, where both problems are considered under the metric hypothesis.
\end{lemma}
\begin{proof}
 The proof is similar to the proof of Lemma \ref{th:reduc-metric-TMC-CFL}.
 We will point out the differences in each paragraph.
 
 In the first paragraph, we do not set a capacity limit for the dummy facilities.
 The second paragraph is identical.
 In the third paragraph, we can build $(\hat{x}^2,\hat{z}^2)$ such that each dummy facility $i_0$ fully serves its corresponding client $j_0$, because it has unlimited capacity.
 We can additionally make sure that no other client $j_1$ is served by $i_0$ by construction of the cost $c_{i_0 j_1}$: there exists a non-dummy facility $i_1$ at distance at most $c_{i_0 j_1}$ from $j_1$.
 Therefore, if $i_0$ sends a (non-zero) flow to $j_1$, set this flow to $0$ and send it instead from $i_1$, at no additional cost.
\end{proof}

\begin{theorem}
 There is a $1.488$-approximation for the metric UTMC.
\end{theorem}
\begin{proof}
 The result follows directly from Lemma \ref{th:reduc-metric-UTMC-UFL} and from the existence of a $1.488$-approximation algorithm for UFL in the metric case (Theorem 1 of \cite{li13a}).
\end{proof}

\subsection{The Metric Uncapacitated Facility Location problem with Market Choice}\label{sec:UFLMC}
UFLMC is also known as the Uncapacitated Facility Location with (Linear) Penalties.
In this variant, both opening costs for facilities and penalties for clients that are not served are considered.
The best approximation algorithm for this problem is due to Li et al \cite{li13b}, and has a performance ratio of 1.5148.

Since the currently best-known approximation ratio (1.488 \cite{li13a}) for UFL is better than the one for UFLMC, we have tried reducing UFLMC to UFL.
By Lemma \ref{th:reduc-metric-TMC-CFL} we know this reduction works when both problems are capacitated.
However, without capacities on dummy vertices (as in Lemma \ref{th:reduc-metric-TMC-CFL}) or non-dummy facilities that do not have an opening cost (as in Lemma \ref{th:reduc-metric-UTMC-UFL}), the approximation-preserving reductions used throughout this paper do not seem to carry over to this variant.

\section*{Acknowledgements}
The authors would like to thank Marco Molinaro and George Nemhauser for inviting Karen Aardal to the Industrial and Systems Engineering department of the Georgia Institute of Technology and making this work possible.
Pierre Le Bodic's research was funded by AFOSR grant FA9550-12-1-0151 of the Air Force Office of Scientific Research and the National Science Foundation Grant CCF-1415460 to the Georgia Institute of Technology.

\bibliography{tmc}
\bibliographystyle{plain}
\end{document}